\documentclass[12pt, a4paper]{article}

\usepackage{amsmath,amsfonts,amsthm,amssymb}
\usepackage[retainorgcmds]{IEEEtrantools}
\usepackage{fourier}
\usepackage{tikz}

\usepackage{enumitem}
\usepackage{booktabs}
\usepackage{color}

\theoremstyle{plain}
\newtheorem{theorem}{Theorem}[section]

\newtheorem{lemma}[theorem]{Lemma}

\newtheorem{definition}[theorem]{Definition}
\newtheorem{corollary}[theorem]{Corollary}

\theoremstyle{definition}
\newtheorem{example}{Example}

\newcommand{\mms}{\boldsymbol{\mu}}
\newcommand{\mysetminusD}{\hbox{\tikz{\draw[line width=0.6pt,line cap=round] (3pt,0) -- (0,6pt);}}}
\newcommand{\mysetminusT}{\mysetminusD}
\newcommand{\mysetminusS}{\hbox{\tikz{\draw[line width=0.45pt,line cap=round] (2pt,0) -- (0,4pt);}}}
\newcommand{\mysetminusSS}{\hbox{\tikz{\draw[line width=0.4pt,line cap=round] (1.5pt,0) -- (0,3pt);}}}
\newcommand{\mysetminus}{\mathbin{\mathchoice{\mysetminusD}{\mysetminusT}{\mysetminusS}{\mysetminusSS}}}

\allowdisplaybreaks

\title{On Truthful Mechanisms for \\ Maximin Share Allocations\thanks{A preliminary conference version appeared in IJCAI 2016.}}
\author{
Georgios Amanatidis 
\and
Georgios Birmpas 
\and
Evangelos Markakis \and 
\normalsize{Athens University of Economics and Business, Department of Informatics.} \\
\normalsize{\{gamana, gebirbas, markakis\}@aueb.gr}
}

\begin{document}

\maketitle

\begin{abstract}
We study a fair division problem with indivisible items, namely the computation of maximin share allocations. 
Given a set of $n$ players, the maximin share of a single player is the best she can guarantee to herself, if she
would partition the items in any way she prefers, into $n$ bundles,
and then receive her least desirable bundle. The objective then is
to find an allocation, so that each player is guaranteed her maximin share.

Previous works have studied this problem mostly algorithmically, providing constant factor approximation algorithms. In this work we embark on a mechanism design approach and investigate the existence of truthful mechanisms. We propose three models regarding the information that the mechanism attempts to elicit from the players, based on the cardinal and ordinal representation of preferences. We establish positive and negative (impossibility) results for each model and highlight the limitations imposed by truthfulness on the approximability of the problem.
Finally, we pay particular attention to the case of two players, which already leads to challenging questions.
\end{abstract}

\section{Introduction}
\label{sec:intro}

We study the design of mechanisms for a fair division problem with indivisible items. 
The objective in fair division is to allocate a set of resources to a set of
players in a way that leaves everyone satisfied, according to their own preferences.
Over the past decades, several fairness concepts have been proposed and the area gradually
gained popularity in computer science as well, since most of the questions are inherently
algorithmic. We refer the reader to the upcoming survey \cite{Procaccia14-survey} for more recent results and to the classic textbooks \cite{BT96,RW98} for an overview of the area.

Our focus here is on the concept of maximin share allocations, which has already attracted a lot of attention ever since it was introduced by Budish \cite{Budish11}.
The rationale for this notion is as follows: suppose that a player, say player $i$, is asked to partition
the items into $n$ bundles and then the rest of the players select a bundle before $i$. In
the worst case, player $i$ will be left with her least valuable subset. Hence, a risk-averse
player would choose a partition that maximizes the minimum value of a bundle. 
This value is called the maximin share of agent $i$ and the goal then is to
find an allocation where every person receives at least her maximin share.  

The existence of maximin share allocations is not always guaranteed under indivisible items. This has led to a series of works that have either established approximation algorithms (i.e., every player receives an approximation of her own maximin share) or have resolved special cases of the problem; see our Related Work section. Currently, the best algorithms we are aware of achieve an approximation ratio of $2/3$ \cite{PW14,AMNS15}, and it is still a challenging open problem if one can do better. 

These previous works, apart from examining the existence of maximin share allocations, have studied the problem from an algorithmic point of view, and one aspect that has not been addressed so far is incentive compatibility. Players may have incentives to misreport their valuation functions and in fact, the proposed approximation algorithms are not truthful. Is it possible then to have truthful algorithms with the same approximation guarantee? Truthfulness is a demanding constraint especially in settings without monetary transfers, and our goal is to explore the effects on the approximability of the problem as we impose such a constraint.   

\vspace{8 pt}\noindent {\bf Contribution.}
We investigate the existence of truthful deterministic mechanisms for constructing approximate or exact maximin share allocations. In doing so, we consider three models regarding the information that the mechanism attempts to elicit from the players. 

The first one is the more straightforward approach where players have to submit their entire additive valuation function to the mechanism. We then move to mechanisms where the manipulating power of the players is restricted by the type of information that they are allowed to submit. Namely, in our second model players only submit their ranking over the items, motivated by the fact that many mechanisms in the fair division literature fall within this class. Finally, in our third model we assume the mechanism designer knows the ranking of each player over the items and asks for a valuation function consistent with the ranking. This can be appropriate for settings where the items are distinct enough to extract a ranking, or when the players are known to belong to specific behavioral types.   
For each of these models, we establish positive and negative (impossibility) results and highlight the differences and similarities between them.  
Our results provide a clear separation between the guarantees achievable by truthful and non-truthful mechanisms. We also note that all our positive results yield polynomial time algorithms, whereas the impossibility results are independent of the running time of an algorithm.  
Moreover, we pay particular attention to the case of two players, which already gives rise to non-trivial questions, even with a small number of items. 

Finally, motivated by the lack of positive results for deterministic mechanisms, we analyze the performance of a very simple truthful randomized mechanism.
For a wide range of distributions for the values of the items, we show that we have an arbitrarily good approximation when the number of the items is large enough.

 


\vspace{8 pt}\noindent {\bf Related Work.} The notion of maximin share allocations was introduced by Budish \cite{Budish11} (building on concepts of Moulin \cite{Moulin90}), and later on defined by Bouveret and Lema{\^i}tre \cite{BL14} in the setting that we study here. Both experimental and theoretical evidence, see \cite{BL14,KPW16,AMNS15}, indicate that such allocations do exist almost always. As for computation, a $2/3$ approximation algorithm has been established by Procaccia and Wang \cite{PW14} and later on, a polynomial time algorithm with the same guarantee was provided by Amanatidis et al.~\cite{AMNS15}.  

Regarding incentive compatibility, we are not aware of any prior work that addresses the design of truthful mechanisms for maximin share allocations. There have been quite a few works on mechanisms for other fairness notions, see among others \cite{CKKK09,CLPP13,LMMS}. Parts of our work are motivated by the question of what is the power of cardinal information versus ordinal information. We note that exploring what can be done using only ordinal information has been recently studied for other optimization problems too, (see \cite{AS16}). A popular class of mechanisms based only on ordinal preferences is the class induced by ``picking sequences'', introduced by Kohler and Chandrasekaran \cite{KC71}; see also references in Section \ref{sec:rankings}. We  make use of such algorithms to establish some of our positive results.   

Finally, we should note that it has been customary in the context of fair division to not allow side payments, hence these are mechanism design problems without money. We stick to the same approach here.

\section{Preliminaries}
\label{sec:prelim}
For any $k\in\mathbb N$, we denote by $[k]$ the set $\{1,\ldots ,k\}$.
Let $N = [n]$ be a set of $n$ players and $M = [m]$ be a set of indivisible items.
We assume each player has an additive valuation function $v_i(\cdot)$ over the items, and we will write $v_{ij}$ instead of $v_i(\{j\})$. For $S\subseteq M$, we let $v_i(S) = \sum_{j\in S} v_{ij}$. 
An allocation of $M$ to the $n$ players is a partition,
$T = (T_1,...,T_n)$, where $T_i\cap T_j = \emptyset$ and $\bigcup_i T_i = M$.
Let $\Pi_n(M)$ be the set of all partitions of a set $M$ into $n$ bundles.

\begin{definition}
Given $n$ players, and a set  $M$ of items, the $n$-maximin share of a player $i$ with respect to $M$ is:
\[ \mms_i(n, M) = \displaystyle\max_{T\in\Pi_n(M)} \min_{T_j\in T} v_i(T_j)\,.\]
\end{definition}

When it is clear from the context what $n, M$ are, we will simply write $\mms_i$ instead of $\mms_i(n, M)$. 
The solution concept defined in \cite{Budish11} asks for a partition that gives each player her maximin share.
\begin{definition}
Given $n$ players and a set of items $M$, a partition 
$T = (T_1,...,T_n) \in \Pi_n(M)$ is called a maximin share allocation 
if $v_i(T_i)\geq \mms_i\,$, for every $i\in [n]$. If $v_i(T_i)\geq \rho\cdot\mms_i\,$, $\forall i\in [n]$, with $\rho\leq 1$, 
then $T$ is called a $\rho$-approximate maximin share allocation.
\end{definition} 
It can be easily seen that this is a relaxation of the classic notion of proportionality. 

\begin{example}
\label{ex:mms}
Consider an instance with 3 players and 5 items: 
\begin{center}
\begin{tabular}{@{} *6l @{}}    
 & $a$ & $b$ & $c$ & $d$ & $e$ \\\toprule
\ Player 1 & $1/2$ & $1/2$ & $1/3$ & $1/3$ & $1/3\ $ \\ 
\ Player 2 & $1/2$ & $1/4$ & $1/4$ & $1/4$ & $0\ $ \\ 
\ Player 3 & $1/2$ & $1/2$ & 1 & $1/2$ & $1/2\ $ \\\bottomrule
\end{tabular}
\end{center}\vspace{5pt}
%
If $M = \{a, b, c, d, e\}$ is the set of items, one can see that $\mms_1(3, M) = 1/2$, $\mms_2(3, M) = 1/4$, $\mms_3(3, M) = 1$. For player $1$, no matter how she partitions the items into three bundles, the worst bundle will be worth at most $1/2$ for her. 
Similarly, player $3$ can guarantee a value of $1$ (which is best possible as it is equal to $v_3(M)/n$) by the partition $(\{a, b\}, \{c\}, \{d, e\})$. Note that this instance admits a maximin share allocation, e.g., $(\{a\}, \{b, c\}, \{d, e\})$, and in fact this is not unique.

Note also that if we remove some player, say player 2, the maximin values for the other two players increase. E.g., $\mms_1(2, M) = 1$, achieved by the partition $(\{a, b\}, \{c, d, e\})$. Similarly, $\mms_3(2, M) = 3/2$.  
\end{example}

\subsection{Mechanism design aspects} 

Following most of the fair division literature, our focus is on mechanism design without money, i.e., we do not allow side payments to the players. Then, the standard way to define truthfulness  is as follows:  
an instance of the problem can be described as an $n\times m$ valuation matrix $V=[v_{ij}]$, as in Example \ref{ex:mms} above. 
For any mechanism $A$, we denote by $A(V)=(A_1(V), \ldots, A_n(V))$ the allocation output by $A$ on input $V$.
Also, let $\mathbf v_i$ denote the $i$th row of $V$, and $V_{-i}$ denote the remaining matrix. 
Finally, let $(\mathbf v_i', V_{-i})$ be the matrix we get by changing the $i$th 
row of $V$ from $\mathbf v_i$ to $\mathbf v_i'$. 

\begin{definition}
\label{def:truthful}
A mechanism $A$ is truthful if for any instance $V$, any player $i$, and any  possible declaration $\mathbf v_i'$ of  $i$:
$v_i(A_i(V)) \ge v_i(A_i(\mathbf v_i', V_{-i}))$.
\end{definition}

Obtaining a good understanding of truthful mechanisms and their performance for other fairness notions has been a difficult problem; see among others \cite{LMMS,CKKK09} for approximating minimum envy with truthful mechanisms.
The difficulty is that an algorithm that uses an $m$-dimen\-sion\-al vector of values for each player, can create many subtle ways for players to benefit by misreporting. One can try to alleviate this by restricting the type of information that is requested from the players.  As a first instantiation of this, we note that many mechanisms in the literature end up utilizing only the ranking of each player for the items, and not the entire valuation function, (see our discussion in Section \ref{sec:rankings} and references therein). This yields simpler, intuitive mechanisms, at the expense of possibly sacrificing performance, since the mechanism uses less information.  
As a second instantiation, one can exploit information that could be available to the mechanism so as to restrict the allowed valuations. For example, in some scenarios, it is realistic to 
assume that the ranking of each player over the items is public knowledge. If the items are distinct enough, it is possible that one could extract such information (a special case is that of full correlation, considered in \cite{BF02,BL11}, where all players have the same ranking). Therefore, the players in such cases can only submit values that agree with their (known) ranking.

Motivated by the above considerations, we study the following three models: 

\begin{itemize}[leftmargin=20pt] 
\item The Cardinal or Standard Model. Every player submits a valuation function, without any restrictions. To represent the input of player $i$, we fix an ordering of the items and write the corresponding vector of values as $\mathbf v_i=[v_{i1}, \allowbreak v_{i2}, \ldots, v_{im}]$.
\item The Ordinal Model. Here, an instance is again determined by a matrix $V$, however a mechanism only asks players to submit a ranking on the items. 
Note that Definition \ref{def:truthful} of truthfulness needs to be modified accordingly. That is, 
let $\succeq_i$ be any total order consistent with $\mathbf v_i$ (there may be many in case of ties). 
A mechanism  is truthful if for any tuple of rankings for the other players, denoted by $\succeq_{-i}$, and any ranking $\succeq_i'\,$: $v_i(A_i(\succeq_i, \succeq_{-i})) \ge v_i(A_i(\succeq_i', \succeq_{-i}))$.
\item The Public Rankings Model. Now, the ranking of each player is known to the mechanism, say it is $\succeq_i$. Hence, each player is asked to submit a valuation function consistent with $\succeq_i$.   
\end{itemize}

It is not hard to see how the different scenarios we investigate are related to 
each other; this is summarized in the following lemma. 

\begin{lemma}\label{lem:implications}
\emph{(i)} Assume there exists a truthful $\rho$-approximation mechanism $A$ in the cardinal model. 
Then, $A$ can be efficiently turned into a truthful $\rho$-approximation 
mechanism for the public rankings model.\\
\emph{(ii)} Assume there exists a truthful $\rho$-approximation mechanism $A$ for the ordinal model.  
Then, $A$ can be efficiently turned into a truthful $\rho$-approximation 
mechanism for the cardinal model.
\end{lemma}

\begin{proof}
\emph{(i)} The mechanism $B$ for the public rankings model will take as input the values ${\mathbf v}'_i=[v'_{i1}, \ldots, v'_{im}], \forall i\in[n]$ as reported by the players, together with the actual (publicly known) rankings $\succeq_i, \forall i\in[n]$. If for some $j$ the reported ${\mathbf v}'_j$ is not consistent  with $\succeq_j$, then player $j$ is ignored by the mechanism. For all the consistent players $B$ runs $A$ on the same inputs and outputs the same allocation as $A$.
Clearly, no player has incentive to be inconsistent with her ranking. Given that, the truthfulness of $B$ follows from the truthfulness of $A$, as does the approximation ratio. \\
\emph{(ii)} The mechanism $B$ for the cardinal model will take as input the values ${\mathbf v}'_i=[v'_{i1}, \ldots, v'_{im}], \allowbreak \forall i\in[n]$ as reported by the players, and will produce the corresponding rankings $\succeq'_i, \forall i\in[n]$. Then, $B$ runs $A$ using as input the rankings $\succeq'_i, \forall i\in[n]$ and outputs the same allocation as $A$.
It is clear that no player has incentive to misreport her values without changing her actual ranking. Given that, the truthfulness of $B$ follows from the truthfulness of $A$, as does the approximation ratio.
\end{proof}

\section{The Cardinal Model}
\label{sec:values}

As already alluded to, designing mechanisms that utilize the values submitted by each player, so as to achieve a good approximation and at the same time induce truthful behavior, is a very challenging problem.
This is true even in the case of $n=2$ players. Therefore, we start first with a rather weak result for general $n$ and $m$, and then move on to discuss the case of two players.
The main message from this section (Theorem \ref{thm:val-1/2}) is that there is a clear separation, regarding the approximation guarantees of truthful and non-truthful algorithms. 

\begin{theorem}\label{cor:val-gen}
For any $n\ge 2, m\ge 1$, there is a truthful 
$1/\left\lfloor \frac{\max\{2,  m-n+2\}}{2}\right\rfloor$-approximation 
mechanism for the cardinal model.
\end{theorem}
The proof follows from results in the next section (see the discussion before and after Theorem \ref{thm:tr_pick_sec}). For the case of two players, the mechanism of Theorem \ref{cor:val-gen}  has the following form:\vspace{-3pt}\\
\rule{\columnwidth}{0.6pt}
\noindent Mechanism $M_{\textsc{best item}}$: Given the reported valuations of the players, 
allocate to player 1 her best item and to player 2 the remaining items.\vspace{-3pt}\\
\rule{\columnwidth}{0.6pt}

Although the approximation ratio achieved by Theorem \ref{cor:val-gen} is quite small,
it is still an open question whether there exist better mechanisms for general $n, m$. We note also that $M_{\textsc{best item}}$ only utilizes the preference rankings of the players. Hence it is not even clear if there exist truthful mechanisms that can exploit more information from the valuation functions to achieve a better approximation.
 
For the remainder of this section, we discuss the case of $n=2$. 
We recall that for two players, the discretized cut and choose procedure is a non-truthful algorithm that produces an exact maximin share allocation; one player partitions the goods into two bundles that are as equal as possible, and the other player chooses her best bundle. To implement this in polynomial time, we can produce an approximate partitioning using a result of Woeginger \cite{Woeginger97} and then we can guarantee at least $(1-\varepsilon)\mms_i$ to 
each player, $\forall \varepsilon>0$. The reason this is not truthful is that player $1$ can manipulate the partitioning; in fact, she can compute her optimal strategy if she knows the valuations of player $2$ by solving a Knapsack instance. Thus, the question we would like to resolve is to find the best truthful approximation that we can guarantee for two players.

Notice that for $n=2, m<4$, the mechanism $M_{\textsc{best item}}$
does output an exact maximin share allocation. 
Further, when $m\in\{4, 5\}$, $M_{\textsc{best item}}$ outputs a $1/2$-approximation, according to Theorem \ref{cor:val-gen}.

On the other hand, we can deduce an impossibility result, using Theorem 5 of \cite{MP11}, which yields\footnote{The work of \cite{MP11} concerns a different problem however the arguments for their impossibility result can be employed here as well.}:

\begin{corollary}[implied by \cite{MP11}]
For $n=2, m\ge 4$, and for any $\varepsilon\in (0, 1/3]$, there is no truthful 
$(2/3 + \varepsilon)$-approximation mechanism for the cardinal model.
\end{corollary}


The above corollary leaves open whether there exist better mechanisms than $M_{\textsc{best item}}$ with approximation guarantees in $(1/2, 2/3]$. Our main result in this section is that we close this gap, by providing a stronger negative result, which shows that $M_{\textsc{best item}}$  is optimal for $n=2, m=4$. 
\begin{theorem}\label{thm:val-1/2}
For $n=2, m\ge4$, and for any $\varepsilon\in (0, 1/2]$, there is no truthful 
$(1/2 + \varepsilon)$-approximation mechanism for the cardinal model.
\end{theorem}

We prove the theorem for $m=4$, since we can trivially extend it to any number of items by adding dummy items of no value. The proof follows from Lemmas \ref{lem:val-1/2-2+2} and 
\ref{lem:val-1/2-1+3} below. Notice that the theorem is valid even if we drastically 
restrict the possible values of the items. 

\begin{lemma}\label{lem:val-1/2-2+2}
For $n=2, m=4$, and for any $\varepsilon\in (0, 1/2]$, there is no truthful $(1/2 + \varepsilon)$-approximation 
mechanism for the cardinal model that allocates two items to each player at every instance where the 
profiles are permutations of $\{2+\varepsilon, 1+\varepsilon, 1-\varepsilon, \varepsilon /2\}$ 
or $\{2-\varepsilon, 1+\varepsilon, 1-\varepsilon, \varepsilon /2\}$.
\end{lemma}

\begin{proof}Let us first fix an ordering of the four items, say $a, b, c, d$. For the sake of readability we write $2^{+}, 2^{-}, 1^{+}, 1^{-},\allowbreak 0^{+}$ 
instead of $2+\varepsilon, 2-\varepsilon,  1+\varepsilon, 1-\varepsilon$ and $\varepsilon /2$.

Suppose that there is such a mechanism. We use six different profiles, 
the values of which are permutations of either $\{2^{+}, 1^{+}, 1^{-}, 0^{+}\}$ 
or $\{2^{+}, 1^{+}, 1^{-}, 0^{+}\}$. Notice that in such profiles the maximin share is $2+\varepsilon/2$ 
or $2-\varepsilon/2$ respectively. Since we want allocations that give to each player items of value 
at least $1/2 + \varepsilon$ of their maximin share, it is trivial to check that allocating $\{1^+, 0^+\}$ or $\{1^-, 0^+\}$ to a player is not feasible (we use this repeatedly below). 
The goal is to get a contradiction by reaching a profile where no possible allocations exist.

\noindent\textit{Profile 1:} $\{[2^-,1^+,1^-,0^+], [2^-,1^+,1^-,0^+]\}$. There are two feasible allocations, i) $(\{a,d\}, \allowbreak \{b,c\})$  and  ii)  $(\{b,c\}, \{a,d\})$.  W.l.o.g. we may assume that the mechanism outputs allocation i). The analysis in the other case is symmetric.


\noindent\textit{Profile 2:} $\{[2^-,1^+,1^-,0^+], [0^+,2^+,1^-,1^+]\}$. There are three feasible allocations,
 i) $(\{a,c\}, \allowbreak \{b,d\})$,  ii)  $(\{a,d\}, \{b,c\})$ and iii)  $(\{a,b\}, \{c,d\})$.  Allocation iii) is not possible, since 
 $p_2$ here could play $\mathbf v_2'=[2^-,1^+,1^-,0^+]$ like in Profile 1 and get a total value of $3>2$. 
 Allocations i) and ii) are currently possible.

\noindent\textit{Profile 3:} $\{[1^+,2^-,1^-,0^+], [0^+,2^+,1^-,1^+]\}$. There are two feasible allocations,  i) $(\{a,c\}, \allowbreak \{b,d\})$  and  ii)  $(\{a,b\}, \{c,d\})$. If the mechanism given Profile 2 outputs allocation ii), then neither allocation here is possible. Indeed, 
in Profile 2
$p_1$ could play $\mathbf v_1'=[1^+,2^-,1^-,0^+]$ like here
and get a total value of $3-2\varepsilon>2-\varepsilon/2$ or $3>2-\varepsilon/2$. Thus,  allocation ii) at Profile 2 is not possible. So, the mechanism, given Profile 2, outputs allocation i) of that profile. Then, using the same argument, allocation ii) 
here is not possible, since in Profile 2
$p_1$ could play $\mathbf v_1'=[1^+,2^-,1^-,0^+]$ like here and get a total value of $3>3-\varepsilon$. So, here, the mechanism outputs allocation i). 


\noindent\textit{Profile 4:} $\{[1^+,2^-,1^-,0^+], [1^+,2^+,1^-,0^+]\}$.  There are two feasible allocations,  i) $(\{a,c\}, \allowbreak \{b,d\})$  and  ii)  $(\{b,d\}, \{a,c\})$. Allocation ii) is not possible, since 
$p_2$ here could play $\mathbf v_2'=[0^+,2^+, \allowbreak 1^-,1^+]$ like in Profile 3
and get a total value of $2+3\varepsilon/2>2$. Thus, the mechanism outputs allocation i).

\noindent\textit{Profile 5:} $\{[1^+,1^-,2^-,0^+], [2^-,1^+,1^-,0^+]\}$.  There are two feasible allocations,  i) $(\{c,d\}, \allowbreak \{a,b\})$  and  ii)  $(\{b,c\}, \{a,d\})$. Allocation ii) is not possible, since in Profile 1
$p_1$ could play $\mathbf v_1'= \allowbreak [1^+,1^-,2^-,0^+]$  like here and get a total value of $2>2-\varepsilon/2$.
Thus, the mechanism outputs allocation i).

\noindent\textit{Profile 6:} $\{[1^+,1^-,2^-,0^+], [1^+,2^+,1^-,0^+]\}$.  There are two feasible allocations,  i) $(\{c,d\}, \allowbreak \{a,b\})$  and  ii)  $(\{a,c\}, \{b,d\})$. Allocation ii) is not possible, since 
$p_2$ here could play $\mathbf v_2'=[2^-, 1^+, \allowbreak 1^-, 0^+]$ like in Profile 5 
and get a total value of $3+2\varepsilon>2+3\varepsilon/2$. However, allocation i) is not possible either, since 
$p_1$ here could play $\mathbf v_1'=[1^+,2^-,1^-,0^+]$ like in Profile 4
 and get a total value of $3>2-\varepsilon/2$. 
 So, we can conclude that there are no possible allocations for this profile, which is a contradiction.
\end{proof}

Using Lemma \ref{lem:val-1/2-2+2}, we deduce that there must exist some instance where the mechanism allocates one item to one player and three items to the other. We prove below that this is not possible either.  

\begin{lemma}\label{lem:val-1/2-1+3}
For $n=2, m=4$, and for any $\varepsilon\in (0, 1/2]$, there is no truthful $(1/2 + \varepsilon)$-approximation 
mechanism for the cardinal model, which at some instance where the 
profiles are permutations of $\{2+\varepsilon, 1+\varepsilon, 1-\varepsilon, \varepsilon /2\}$ 
or $\{2-\varepsilon, 1+\varepsilon, 1-\varepsilon, \varepsilon /2\}$, allocates exactly one item to one of the players.
\end{lemma}

\begin{proof}
Fix an ordering of the four items, say $a, b, c, d$.  We write $2^{+}, 2^{-}, 1^{+}, 1^{-},\allowbreak 0^{+}$ 
instead of $2+\varepsilon, 2-\varepsilon,  1+\varepsilon, 1-\varepsilon$ and $\varepsilon /2$. Suppose that there is such a truthful mechanism, and an instance  
$\{[v_{1a}, v_{1b}, v_{1c}, v_{1d}], \allowbreak  [v_{2a}, v_{2b}, \allowbreak v_{2c},v_{2d}]\}$ (that we refer to as the \textit{initial profile}), where the mechanism gives one item to $p_1$ 
and three items to $p_2$ (the symmetric case can be handled in the same manner). 
Like in the proof of Lemma \ref{lem:val-1/2-2+2}, it is straightforward to check that allocating $\{1^+, 0^+\}$ or $\{1^-, 0^+\}$ to a player is not feasible.

Recall that the values of each player are a permutation of  either $\{2^{+}, 1^{+}, 1^{-}, 0^{+}\}$ 
or $\{2^{-}, 1^{+}, 1^{-}, 0^{+}\}$. Since $p_1$ gets only one item, its value must be 
$2^+$ or $2^-$. W.l.o.g. we may assume that this item is $a$, so the produced allocation is 
$(\{a\}, \{b,c,d\})$. We will now construct a chain of profiles (Profiles 1-4) which will help us establish a contradiction.

\noindent\textit{Profile 1:} $\{[v_{1a}, v_{1b}, v_{1c}, v_{1d}], [2^-, v_{1b}, v_{1c}, \allowbreak v_{1d}]\}$. It is easy to see that $p_2$ can not get just item $a$, or item $a$ and the item that has value $0^+$, or any proper subset of $\{b,c,d\}$, since 
she could then play $\mathbf [v_{2a}, v_{2b}, v_{2c}, v_{2d}]$ as in the initial profile,
and end up strictly better. Moreover, $p_2$ cannot get a bundle that contains $a$ and (at least) one item with value $1^-$ or $1^+$, because then there is not enough value left for $p_1$. Thus, the only feasible allocation here is $(\{a\}, \{b,c,d\})$.
W.l.o.g., by possibly renaming items $b,c,d$, we take Profile 1 to be $\{[v_{1a},1^+,1^-,0^+]$, $[2^-, 1^+,1^-,0^+]\}$.

\noindent\textit{Profile 2:} $\{[v_{1a},1^+,1^-,0^+],[0^+,2^-,1^+,1^-]\}$. It is easy to notice that in any feasible allocation other than $(\{a\}, \{b,c,d\})$, $p_2$ could play $\mathbf v_2'=[2^-, 1^+,1^-,0^+]$ 
as in Profile 1 and end up with a better value. Thus, the mechanism here has to output $(\{a\}, \{b,c,d\})$. 


\noindent\textit{Profile 3:} $\{[1^-,v_{1a},0^+,1^-],[0^+,2^-,1^+,1^-]\}$. Here, $p_1$ cannot get a proper superset of $\{a\}$, since then in Profile 1 she could  play $\mathbf v_1'=[1^-,v_{1a},0^+,1^-]$ 
like here and end up  strictly better. The only other feasible allocation here is $(\{b\}, \{a,c,d\})$.

\noindent\textit{Profile 4:} $\{[1^-,v_{1a},0^+,1^-],[2^-,1^+,1^-,0^+]\}$. Here, $p_2$ cannot get $\{b, c\}$ or any proper subset of $\{a,c,d\}$, since she could then play $\mathbf v_2'=[0^+,2^-,1^+,1^-]$ 
like in Profile 3 and end up with a total value of $3-3\varepsilon/2$, which is strictly better. The only other feasible allocation here is $(\{b\}, \{a,c,d\})$.

By starting now at Profile 2 and repeating the arguments for Profiles 1, 2, and 3 --shifted one position to the right-- we have that for \textit{Profile 5:} $\{[1^-, v_{1a},0^+,1^+], \allowbreak [1^-,0^+,2^-,1^+]\}$ the only possible allocation is $(\{b\}, \{a,c,d\})$, and for \textit{Profile 6:} $\{[1^+,1^-, v_{1a},0^+],[1^-,0^+,2^-,1^+]\}$ the only possible allocation is $(\{c\}, \{a,b,d\})$.

\noindent\textit{Profile 7:} $\{[1^+,1^-, v_{1a},0^+],[2^-,1^+,1^-,0^+]\}$. Here, $p_2$ cannot receive $\{b, c\}$ or any proper subset of $\{a,c,d\}$, since she could then play $\mathbf v_2'=[1^-,0^+,2^-,1^+]$ 
as in Profile $6$ and be better off. The only other feasible allocation is $(\{c\},\allowbreak \{a,b,d\})$.

\noindent\textit{Final profile:} $\{[1,1,1,1],[2^-,1^+,1^-,0^+]\}$. Here, any feasible allocation has to give $p_1$  at least two items, otherwise it is not a $(1/2 + \varepsilon)$-approximation. However, one can check that for any such allocation, there is a profile among Profiles $1$, $4$ and $7$, where $p_1$ could play $\mathbf v_1'=[1,1,1,1]$ and end up strictly better. Thus, we conclude that there are no possible allocations here, arriving at a contradiction.
\end{proof}

This concludes the proof of Theorem \ref{thm:val-1/2}.

\section{The Ordinal Model}
\label{sec:rankings}
Several works in the fair division literature have proposed mechanisms that only ask for the ordinal preferences of the players. 
There are various reasons for such assumptions; apart from their simplicity in implementing them, the players themselves may feel more at ease as they may be reluctant to fully reveal their valuation. Here, one extra motive is to restrict the players' ability to manipulate the outcome.
 
A class of such simple and intuitive mechanisms that has been studied in previous works is the class of 
\emph{picking sequence mechanisms}, see, e.g., \cite{KC71,BK05,BT00,BL11,BoLa14,KNW13,KNWX13} and references therein. 
A picking sequence $\pi = p_{i_1}p_{i_2}\ldots p_{i_k}$ is just a sequence of players (possibly with repetitions). Each picking sequence naturally induces a deterministic allocation mechanism for the ordinal model as follows: first give to player $p_{i_1}$ her favorite item, then give to $p_{i_2}$ her favorite among the remaining items, and so on, and keep cycling through $\pi$ until all the items are allocated. Sometimes, periodicity is absent, because the length of the given sequence is at least $m$.

Notice that these mechanisms can be implemented by asking each player for her ranking over the items. 
And note also that these mechanisms are not generally truthful, unless they are sequential dictatorships, i.e.,\ they are induced by picking sequences of the form $p_{i_1}^{m_1}p_{i_2}^{m_2}\ldots p_{i_k}^{m_k}$, where $p_{i_1}, p_{i_2}, \ldots,  p_{i_k}$ are all different players and $\sum_i m_i \ge m$ (see \cite{BoLa14}).
 
Given a set of $n$ players $p_1,\ldots, p_n$, we now define the following mechanism:
\vspace{-3pt}\\
\rule{\columnwidth}{0.6pt}
\noindent  $M_{\textsc{pick seq}}^{(n, m)}$ is the mechanism induced by the picking sequence 
$\pi = p_1 p_2 p_3 \ldots p_{n-2}p_{n-1}p_{n}p_{n}\ldots p_{n}$. \vspace{-20pt}\\
\rule{\columnwidth}{0.6pt}

\vspace{2pt}Thus, given that there are enough items available, the first $n-1$ players receive exactly one item, and the last player receives the remaining $m-n+1$ items. This is a truthful mechanism, given the observation above.  
It is easy to see that if $m\le n+1$, then $M_{\textsc{pick seq}}^{(n, m)}$ constructs an exact maximin share allocation. 
For large values of $m$, 
however, the approximation deteriorates fast and we also have a strong impossibility result.

\begin{theorem}\label{thm:tr_pick_sec}
The mechanism $M_{\textsc{pick seq}}^{(n, m)}$ defined above is a truthful $1/\left\lfloor \frac{m-n+2}{2}\right\rfloor$-approximation for the ordinal model, for any $n\ge 2, m\ge n+2$. Moreover, there is no truthful mechanism for the ordinal model, induced by some picking sequence, that achieves a better approximation factor.
\end{theorem}

\begin{proof}
As mentioned above, the only strategyproof picking sequences are the ones of the form 
$s=a_{i_1}^{m_1} a_{i_2}^{m_2} \ldots a_{i_n}^{m_n}$, where $\{a_{i_1}, \ldots a_{i_n}\}=N$ 
and $\sum_{j=1}^{n}m_j = m$ . 
For ease of notation, we may each time rename the players, so that $a_{i_j}=p_j$. 

Let $k\in N$. If $\sum_{j=1}^{k-1}m_j \ge k$, then consider the case where the values are fully correlated, and for player $p_k$ we have $v_{kj}=1$ for $1\le j\le k$ and $v_{kj}=0$ for $k< j\le m$. 
Clearly, she will get a bundle of value $0$, while $\mms_{k}=1$. So, in order to have any guarantee with respect to the maximin share, we must have
$\sum_{j=1}^{k-1}m_j < k$ and $m_k>0$, for all $k\in N$. 
Then, the only possible sequences are of 
the form $a_{i_1} a_{i_2} \ldots a_{i_{n-1}}a_{i_n}^{m-n+1}$ (like the sequence that induces $M_{\textsc{pick seq}}^{(n, m)}$). 

To show that these picking sequences give a $1/\left\lfloor \frac{m-n+2}{2}\right\rfloor$-approximation we need the following lemma of \cite{AMNS15}:
\begin{lemma}[Monotonicity Lemma \cite{AMNS15}]\label{lem:monotonicity}
 Fix a player $i$ and an item $j$. Then for any other player $k\neq i$, it holds that 
 $\mms_k(n-1, M\mysetminus \{j\}) \geq \mms_k(n, M)$.
\end{lemma}

As above, assume 
$s=p_1 p_2 \ldots p_{n-1} p_{n}^{m-n+1}$, and  notice that  player $p_{n}$ always gets items of total value at least $\mms_{n}$. Any other player, $p_k$, gets one item, say of  value $x$. Let $M'$ be the set of available items right before $p_k$ picks. Apply the Monotonicity Lemma $k-1$ times to get $\mms_{k}(n, M)\le \mms_{k}(n-k+1, M')\le \big\lfloor \frac{m-k+1}{n-k+1}\big\rfloor \cdot \max_{j\in M'} v_{kj} = \big\lfloor \frac{m-k+1}{n-k+1}\big\rfloor \cdot x$. Since $\big\lfloor \frac{m-k+1}{n-k+1} \big\rfloor$ is maximized for $k=n-1$, we get the desired approximation ratio.
\end{proof}

Notice that Theorem \ref{thm:tr_pick_sec} combined with Lemma \ref{lem:implications}(ii) imply Theorem \ref{cor:val-gen}. 

Now, we return to the case of two players. For $n=2$, the mechanism $M_{\textsc{pick seq}}^{(2, m)}$ is identical to mechanism $M_{\textsc{best item}}$ defined in Section \ref{sec:values}. Hence, as already pointed out there, this mechanism achieves a $1/2$-approximation for $m\in \{4, 5\}$. 
We can now combine the impossibility result of Theorem \ref{thm:val-1/2} and Lemma \ref{lem:implications}(ii) to conclude that $M_{\textsc{pick seq}}^{(2, m)}$ is optimal for the ordinal model when $m\in\{4, 5\}$.

\begin{corollary}
\label{cor:ordinal-1/2}
For $n=2, m\ge 4$, and for any $\varepsilon\in (0, 1/2]$, there is no truthful $(1/2 + \varepsilon)$-approximation mechanism for the ordinal model.
\end{corollary}

For the sake of completeness, in the Appendix we include a proof of Corollary \ref{cor:ordinal-1/2} that does not depend on the results of Section \ref{sec:values}.

The impossibility results of Theorem \ref{thm:val-1/2} and Corollary \ref{cor:ordinal-1/2} have a surprising consequence.  The mechanism $M_{\textsc{pick seq}}^{(2, m)}$ achieves the best possible approximation both for the cardinal and the ordinal model, for $m\in\{4, 5\}$. Therefore, in these cases, giving the mechanism designer access to more information does not improve the approximation factor at all, when truthfulness is required! 

We conclude this section with a general result on the limitations of the ordinal model. Judging from the case $n=2$, it seems that the lack of good approximation guarantees in the cardinal model is due to the truthfulness requirement. 
Here, however, an additional issue is the lack of information itself. 
Below, we prove an inapproximability result for any
mechanism in the ordinal model, independent of whether it is truthful or not.

\begin{theorem}\label{thm:1/H_n}
For $n\ge 2$, and for any $\varepsilon>0$, there is no  $(1/{H_n} + \varepsilon)$-approximation algorithm, be it truthful or not, for the ordinal model, where $H_n$ is the $n$\textsuperscript{th} harmonic number, with $H_n = \Theta(\ln n)$.
Moreover, for $n=3$, there is no $(1/2 + \varepsilon)$-approximation algorithm for the ordinal model. 
\end{theorem}

\begin{proof}
Let $A$ be an $\alpha$-approximation algorithm for the ordinal model, where $\alpha>0$. 
Consider an instance with large enough $m$, where all the players  agree on the ranking $1\succeq 2\succeq\ldots\succeq m$. 
Let $g_i$ be the best item that player $i$ receives by $A$. We renumber the players, if needed, 
so that if $i<j$ then $g_i < g_j$. We claim that $g_i=i$. To see that, consider player $n$. 
Clearly, by the definition of $g_n$ and the renumbering of the players, we have $g_n \ge n$. 
If $g_n > n$, let $v_{n 1} = \ldots = v_{n n}=1$ and $v_{n,  n+1} = \ldots = v_{n m} = 0$. Then, in such an instance, algorithm $A$ will fail to give an $\alpha$-approximation of $\mms_n$ to player $n$.
It follows that $g_n = n$, and therefore $1=g_1<g_2<\ldots<g_{n-1}<n$, which implies $g_i=i$ for every $i\in [n]$.

Now, for $i\ge 1$, suppose that $v_{i1} = \ldots = v_{i, i - 1}=1$ and $v_{i i} = \ldots = v_{i m} = \frac{1}{m-i+1}$. Observe that $\mms_i = \big\lfloor \frac{m-i+1}{n-i+1} \big\rfloor \frac{1}{m-i+1}$, and algorithm $A$ must give at least 
$\big\lceil  \alpha  \big\lfloor \frac{m-i+1}{n-i+1} \big\rfloor \big\rceil$ items to player $i$. 

Since there are $m$ items in total, we must have 
$\sum_{i=1}^{n} \big\lceil  \alpha  \big\lfloor \frac{m-i+1}{n-i+1} \big\rfloor \big\rceil\le m$. 
It follows that for any $\varepsilon >0$ and large enough $m$
\[\alpha  \le  \frac{m}{\sum_{i=1}^{n}  \big\lfloor \frac{m-i+1}{n-i+1} \big\rfloor} \le \frac{m}{\sum_{i=1}^{n}  \big( \frac{m-i+1}{n-i+1}  - 1\big)} 
= \frac{1}{\left( 1- \frac{n}{m} \right) \sum_{i=1}^{n} \frac{1}{n-i+1}} 
< \frac{1}{H_n} + \varepsilon \,.\]
Especially for $n=3$, assume that $\alpha>1/2$ and consider the same analysis as above with $m=6$. 
We get the  contradiction 
\[6  \ge  \textstyle{\sum_{i=1}^{3} \big\lceil  \alpha  \big\lfloor \frac{7-i}{4-i} \big\rfloor \big\rceil = \left\lceil  \alpha  \left\lfloor \frac{6}{3} \right\rfloor \right\rceil
+ \left\lceil  \alpha  \left\lfloor \frac{5}{2} \right\rfloor \right\rceil
+ \left\lceil  \alpha  \left\lfloor \frac{4}{1} \right\rfloor \right\rceil} 
\ge 2+2+3=7\,. \qedhere\]
\end{proof}

\section{The Public Rankings Model}
\label{sec:val-rank}
When the players' rankings are publicly known, one would expect to achieve better approximation ratios, while still maintaining truthfulness. Indeed, the mechanism now has more information, while the options for manipulation are greatly reduced. In particular, note that any picking sequence induces a truthful mechanism for the public rankings model.

We show that indeed this is the case; the impossibility results we obtain are less severe and we have improvements for the case of more than two players as well.

We focus first on two players. For $m<4$, the mechanism $M_{\textsc{pick seq}}^{(2, m)}$ from Section \ref{sec:rankings} gives an exact solution, like before. However, unlike what happens in the other two scenarios, for $m=4$ we now have a truthful exact mechanism. Before we describe the mechanism, we introduce some useful notation. 
For a player $i$, we will denote by $B_i(k_1, k_2, \ldots k_\ell)$ the set of items that are in the positions $k_1, k_2, \ldots k_\ell$,
of her ranking. E.g., $B_2(2, 4)$ denotes the bundle that 
contains the second and the fourth items in the ranking of player 2. \vspace{-3pt}\\
\rule{\columnwidth}{0.6pt}
\noindent Mechanism $M_{\textsc{pr-exact}}^{(2, 4)}$: Given the reported valuations of the two players $p_1, p_2$, and their actual rankings, consider two cases:\\
--- If their most valuable items are different, allocate the items according to the picking sequence $p_1 p_2 p_2 p_1$.\\
--- Otherwise, give to player 1 her most valuable bundle among $B_1(1)$ and $B_1(2, 3)$, and to player 2 the remaining items.\vspace{-5pt}\\
\rule{\columnwidth}{0.6pt}

\begin{theorem}
Mechanism $M_{\textsc{pr-exact}}^{(2, 4)}$ is truthful and produces an exact maximin share allocation for the public rankings model, for $n= 2, m=4$. 
\end{theorem}

\begin{proof} 
To see why $M_{\textsc{pr-exact}}^{(2, 4)}$ is truthful, note that the players cannot affect which of the two cases of $M_{\textsc{pr-exact}}^{(2, 4)}$ will be employed, since this is defined by the publicly known rankings. In addition, only $p_1$ could strategize, in the case where she agrees with $p_2$ on the most valuable item. However, in that case $M_{\textsc{pr-exact}}^{(2, 4)}$ gives her the best bundle between two choices defined by her ranking, thus there is no incentive to lie about her true values. 

To prove now the guarantee for the maximin share, observe that when the two players disagree on their most valuable item, $p_1$ receives one of $B_1(1, 2)$, $B_1(1, 3)$, or $B_1(1, 4)$, and $p_2$ receives either $B_2(1, 2)$, or $B_2(1, 3)$. Similarly, when they agree on their most valuable item, $p_1$ receives her best bundle among $B_1(1)$ and $B_1(2, 3)$, and $p_2$ receives either a bundle of three items, or one of $B_2(1, 2)$, $B_2(1, 3)$, or $B_2(1, 4)$.  

Consider the seven possible ways  $p_i$ can split the four items into non-empty bundles: $(B_i(1),\allowbreak B_i(2, 3, 4))$, $(B_i(2),\allowbreak B_i(1, 3, 4))$, $(B_i(3),\allowbreak B_i(1, 2, 4))$, $(B_i(4),\allowbreak B_i(1, 2, 3))$, $(B_i(1, 2),\allowbreak B_i(3, 4))$, $(B_i(1, 3),\allowbreak B_i(2, 4))$ and $(B_i(1, 4),\allowbreak B_i(2, 3))$. From the definition of maximin share, in at least one of those, both bundles have value at least $\mms_i$. 

It is easy to see that the total value of $B_i(1, 3)$ (and thus of $B_i(1, 2)$), is always at least $\mms_i$, and the same holds for any bundle that contains three items.
Moreover, we claim that both $v_i(B_i(1, 4))$ and $\max\{v_i(B_i(1)), v_i(B_i(2, 3))\}$ are at least $\mms_i$, which suffice to prove the theorem. Indeed, if $v_i(B_i(1, 4))< \mms_i$ or $\max\{v_i(B_i(1)),\allowbreak v_i(B_i(2, 3))\} < \mms_i$, this implies that each one of $B_i(1)$, $B_i(2)$, $B_i(3)$, $B_i(4)$, $B_i(2, 3)$, $B_i(2, 4)$, and $B_i(3, 4)$ also has value less than $\mms_i$. Thus, none of the possible partitions has both bundles worth at least $\mms_i$, contradicting the definition of maximin share. 
\end{proof}

An interesting question is whether the above can be extended for any number of items. We exhibit below that the answer is no, hence non-truthful algorithms have a strictly better performance under this model as well. However, for general $m$ we provide later on an improved approximation in comparison to the other two settings.

\begin{theorem}\label{thm:5/6}
For $n=2$, and $m = 5$, there is no truthful $(5/6 + \varepsilon)$-approximation mechanism for any $\varepsilon\in (0, 1/6]$, while for $m\ge 6$, there is no truthful $(4/5 + \varepsilon)$-approximation mechanism for any $\varepsilon\in (0, 1/5]$.
\end{theorem}

\begin{proof}
We give the proof for $m= 6$, which can be extended to $m\geq6$, by adding dummy items of no value. The proof for $m=5$ is of similar flavor, albeit more complicated, and is included in the Appendix. 

Suppose that there exists a deterministic truthful mechanism for the public rankings model that achieves a $(4/5 + \varepsilon)$-approximation for some $\varepsilon>0$. We study five profiles where the ranking of the six items is $a\succeq_i b\succeq_i c\succeq_i d\succeq_i e \succeq_i f$ for $i\in\{1, 2\}$, thus it is feasible for both players to move between these profiles in order to increase the value they get. Recall that in our current model a player can strategize using the values of the items, but without changing their publicly known ranking.

\noindent\textit{Profile 1:} $\{[1, 1, 1, 1, 1, 1],[1,1,1,1,1,1]\}$. Here,  $\mms_i =3$ for $i\in\{1, 2\}$, so in order to achieve a better than a $0.8$-approximation, the mechanism must give to each player items of value greater than $0.8\cdot \mms_i = 2.4$. Thus each player has to receive three items. 
W.l.o.g. we may assume that $p_1$ gets item $a$ (the analysis in the other case is symmetric).

\noindent\textit{Profile 2:} $\{[1,0.2,0.2,0.2,0.2,0.2],[1,1,1,1,1,1]\}$. Here, $\mms_1 = 1$ and  $\mms_2 = 3$. The mechanism must give to $p_1$ a total value greater than $0.8 \cdot 1=0.8$ and to $p_2$ a total value greater than $0.8\cdot 3= 2.4$. Notice now that $p_2$ has to get at least three items, and therefore $p_1$ has to get a superset of $\{a\}$. In fact, $p_1$ gets a superset of $\{a\}$ of size three, otherwise she could play $\mathbf v_1'=[1, 1, 1, 1, 1, 1]$ like in Profile 1 and end up strictly better. So, we conclude that  both players get three items each, and $p_1$ gets item $a$.

\noindent\textit{Profile 3:} $\{[1,0.2,0.2,0.2,0.2,0.2],[1,0.2,0.2,0.2,0.2,0.2]\}$. Here,  $\mms_i =1$ for $i\in\{1, 2\}$, so in order to achieve something strictly greater than $0.8\cdot 1= 0.8$, there are only two feasible allocations: 
i) $(\{b, c, d, e, f\}, \{a\})$, and ii) $(\{a\}, \{b, c, d, e, f\})$.
Now, notice that allocation ii) is not possible, since then at Profile 2 $p_2$ could play $\mathbf v_2'=[1,0.2,0.2,0.2,0.2,0.2]$ like here and end up strictly better.
Thus, the mechanism outputs $(\{b, c, d, e, f\}, \{a\})$.

\noindent\textit{Profile 4:} $\{[1,1,1,1,1,1],[1,0.2,0.2,0.2,0.2,0.2]\}$. Here, $\mms_1= 3$ and  $\mms_2= 1$. The mechanism must give to $p_1$ a total value greater than  $0.8 \cdot 3=2.4$ and to $p_2$ a total value greater than $0.8\cdot 1= 0.8$. Notice now that $p_1$ has to get five items, since otherwise she could play $\mathbf v_1'= \allowbreak [1,0.2,0.2,0.2,0.2,0.2]$ like in Profile 2 and end up strictly better. Thus $p_2$ has to get $\{a\}$ to achieve the desired ratio.

\noindent\textit{Profile 5:} $\{[1,1,1,1,1,1],[0.7,0.3,0.25,0.25,0.25,0.25]\}$. Here, $\mms_1= 3$ and  $\mms_2= 1$. The mechanism must give to $p_1$ a total value greater than $0.8 \cdot 3=2.4$ and to $p_2$ a total value greater than $0.8\cdot 1= 0.8$. First, notice that $p_1$ must get at least three items. Moreover, if the mechanism does not give item $a$ to $p_2$, then there is no way for $p_2$ to get total value strictly greater than 0.8 with at most three items. Therefore, $p_2$ has to get a strict superset of $\{a\}$. However, this is not feasible either, since in Profile 4 $p_2$ could play $\mathbf v_2'=[0.7,0.3,0.25,0.25,0.25,0.25]$ like here and end up strictly better. Thus, we conclude that there are no possible allocations here, arriving at a contradiction.
\end{proof}

Exploiting the fact that picking sequences induce truthful mechanisms for the public rankings model, we can get more positive results for two players and any $m$.
\vspace{-3pt}\\
\rule{\columnwidth}{0.6pt}
\noindent  $M_{\textsc{pr}}^{(2, m)}$ is the mechanism for two players induced by the picking sequence 
$p_1 p_2 p_2$. \vspace{-5pt}\\
\rule{\columnwidth}{0.6pt}

\vspace{2pt} We have the following result for $M_{\textsc{pr}}^{(2, m)}$.


\begin{theorem}
\label{thm:M2}
For $n=2$ and any $m\geq 1$,  $M_{\textsc{pr}}^{(2, m)}$  is a truthful $2/3$-approximation mechanism for the public rankings model.
\end{theorem}

Hence, for $n=2$, we have a pretty clear picture on what we can achieve for any $m$, leaving only a small gap, i.e., $[2/3, 4/5]$, between the impossibility result and Theorem \ref{thm:M2}.

We can also obtain constant factor approximations for more than two players, which has been elusive in the other two models. E.g., for $n=3$, we can achieve a $1/2$-approximation. 
In particular, for any $n\geq 2$, and $m\geq 1$:
\vspace{-3pt}\\
\rule{\columnwidth}{0.6pt}
\noindent  $M_{\textsc{pr}}^{(n, m)}$ is the mechanism induced by the picking sequence $p_1 p_2\allowbreak p_3\allowbreak \ldots p_{n-1} p_n p_n$.
\vspace{-5pt}\\
\rule{\columnwidth}{0.6pt}

\begin{theorem}\label{thm:Mn}
For any $n\geq 2$, and any $m\geq 1$, the mechanism $M_{\textsc{pr}}^{(n, m)}$ is a truthful $\frac{2}{n+1}$-approximation mechanism for the public rankings model. 
\end{theorem}

\begin{proof}
Let $M$ be the initial set of items, and consider player $p_i$, where $1 \leq i \leq n-1$, right before she picks her first item. We can think of $p_i$ as the first player of the picking sequence $p_i p_{i+1}\allowbreak \ldots \allowbreak  p_{n-1} p_n p_n p_1 p_2 \allowbreak \ldots \allowbreak p_{i-1}$ on a new set of items $M' \subseteq M$, in which $i-1$ items have been removed. Now, since $p_i$ picks her best item out of every $n+1$ items, it is easy to see that the total value she gets is at least 
\[\frac{\sum_{j \in M'}v_{ij}}{n+1} \geq \frac{(n-i+1)\mms_i(n-i+1, M')}{n+1} \geq \frac{(n-i+1)\mms_i(n, M)}{n+1} \geq \frac{2\mms_i(n, M)}{n+1}\,,\] 
where the first inequality follows directly from the definition of the maximin share, and the second inequality follows from Lemma \ref{lem:monotonicity} (applied $i-1$ times).

In a similar manner, we have that the total value $p_n$ gets is at least $\frac{2(n-n+1)\mms_n(n, M)}{n+1} = \frac{2\mms_n(n, M)}{n+1}$ which concludes the proof.
\end{proof}


Notice that Theorem \ref{thm:M2} is a corollary of Theorem \ref{thm:Mn}. Also, observe that $\frac{2}{n+1}$ is better than the guarantee of Theorem \ref{thm:tr_pick_sec}. However, one can do significantly better, even when restricted to mechanisms induced by picking sequences as shown in Theorem \ref{thm:sqrt(n)} below.
We should mention though, that the picking sequence constructed in the proof has length $m$; an interesting question is whether there exist short picking sequences that significantly improve $\frac{2}{n+1}$. Of course, it remains an open problem to design truthful mechanisms achieving better --or even no-- dependence on $n$.

\begin{theorem}\label{thm:sqrt(n)}
For $\varepsilon >0$ and large enough $n$ and $m$, there exists 
a truthful $\frac{1}{n^{0.5+\varepsilon}}$-approximation mechanism for the public rankings model.
\end{theorem}

\begin{proof}
Let $\alpha=\frac{1}{n^{0.5+\varepsilon}}$. Below, we show that there exists a picking sequence mechanism
that $\alpha$-approximates maximin share allocations. This directly implies the statement of the theorem. 

Focus on player $p_i$. Assume we had a picking sequence 
that gives the $i$th pick to $p_i$ (call this the $0$th pick of $p_i$) and then keeps giving her her $j$th pick on, or before, the  $\left( i+ j \left\lfloor \frac{n-i+1}{\alpha} \right \rfloor\right)$th overall pick. We claim that this way $p_i$ would get a bundle 
$S_i$ with $v_i(S_i)\ge \alpha \mms_i$. To see that, notice that when $p_i$ starts picking, in the worst case her $i-1$ best items are already gone and the total remaining value is at least 
$\sum_{j}v_{ij} - (i-1)\mms_i \ge (n-i+1)\mms_i $.
Then, because of the distribution of her picks, $p_i$ is going to get at least 
$1/\left\lfloor \frac{n-i+1}{\alpha} \right \rfloor$ 
of this value, i.e., at least
${\left\lfloor \frac{n-i+1}{\alpha} \right \rfloor}^{-1} (n-i+1)\mms_i \ge 
\frac{\alpha}{n-i+1} (n-i+1)\mms_i = \alpha \mms_i $.

Next, we describe how to construct a picking sequence that satisfies the above property for all
players. Notice that we are going to give a single picking sequence of length $m$. The main idea 
is that, if $p_i$ is going to be satisfied, we want her $j$th pick to be no later than the 
$\left( i+ j \left\lfloor \frac{n-i+1}{\alpha} \right \rfloor\right)$th overall pick. So, we make 
sure there is not too large demand from other players for picks that come before the 
$\left( i+ j \left\lfloor \frac{n-i+1}{\alpha} \right \rfloor\right)$th. The construction 
itself is very simple:
\begin{itemize}
\item[--] For $1\le i \le n$ and 
$0\le j \le \left\lfloor \frac{\alpha (m-i)}{n-i+1} \right \rfloor$, we create the pair 
$\left(p_i, i+ j \left\lfloor \frac{n-i+1}{\alpha} \right \rfloor\right)$.
\item[--] We sort the pairs with respect to their second coordinate.
\item[--] The first coordinates with respect to the above sorting are a prefix of the picking sequence.
\item[--] If the length of the above sequence is $m$, we are done; otherwise we arbitrarily assign the remaining picks.
\end{itemize}
There are two things to be proven here. The first is to show that the third step of the construction 
does not give a picking sequence of length greater than $m$. This is not hard to see, given that $n, m$ 
are large enough. 
There are at most $\left\lfloor \frac{\alpha (m-i)}{n-i+1} \right \rfloor + 1 $ pairs for each $i$, so by summing up we have:
\begin{IEEEeqnarray*}{rCl}
\sum_{i=1}^{n}\left( \left\lfloor \frac{\alpha (m-i)}{n-i+1} \right \rfloor + 1 \right)  
& \le & n + \sum_{i=1}^{n}\frac{\alpha (m-i)}{n-i+1} 
 \le   n + \alpha m\sum_{i=1}^{n}\frac{1}{n-i+1} -  \alpha \sum_{i=1}^{n}\frac{i}{n-i+1}  \\
& \le &  n + \alpha H_n m \le  m \,. 
\end{IEEEeqnarray*}
(Notice that $n$ and $m$ need not be very large for the last inequality to hold. E.g., for $\varepsilon=0.15$, it suffices to have $n\ge 5, m\ge \frac{n}{1-\alpha H_n}$.)


The second goal here is to show that the resulting picking sequence has 
the desired property, i.e., for any $i, j$ there are at most 
$ i+ j \left\lfloor \frac{n-i+1}{\alpha} \right \rfloor$ pairs that come no later than  
$\left(p_i, i+ j \left\lfloor \frac{n-i+1}{\alpha} \right \rfloor\right)$ in the sorting of the second step.
For fixed $i, \ell$ and $j=0$ we have 
\[  \ell+ k \left\lfloor \frac{n-\ell+1}{\alpha} \right \rfloor
\le i \ \Rightarrow \
\begin{cases}
    k = 0 \,,   & \text{if } \ell\le i\\
    \text{contradiction} \,,  & \text{if } \ell> i\\
\end{cases}\ , \]
and therefore there are exactly $i$ pairs that come no later than 
$\left(p_i, i\right)$ in the sorting. To see the contradiction, notice that for 
$k\ge 1$ we have
\[ \ell+ k \left\lfloor \frac{n-\ell+1}{\alpha} \right\rfloor \le i
\ \Rightarrow \ \ell+ \frac{n-\ell+1}{\alpha} - 1  \le n 
\ \Rightarrow \ \alpha\ge 1\,. \]
Now, for fixed $i, \ell$ and $j\ge 1$ we have 
\[  \ell+ k \left\lfloor \frac{n-\ell+1}{\alpha} \right \rfloor
\le i+ j \left\lfloor \frac{n-i+1}{\alpha} \right \rfloor 
\ \Rightarrow \ k\le \left\lfloor \frac{i-\ell+ j \left\lfloor\frac{n-i+1}{\alpha}\right\rfloor}{\left\lfloor \frac{n-\ell+1}{\alpha}\right\rfloor} \right\rfloor  ,\]
where $k$ can be as small as $0$. 
Therefore, we need to show that
\begin{equation}\label{sumoffloors}
 n + \sum_{\ell=1}^{n}  \left\lfloor \frac{i-\ell+ j \left\lfloor\frac{n-i+1}{\alpha}\right\rfloor}{\left\lfloor \frac{n-\ell+1}{\alpha}\right\rfloor} \right\rfloor  
\le i+ j \left\lfloor \frac{n-i+1}{\alpha}\right\rfloor . 
\end{equation}
Before we prove \eqref{sumoffloors}, we should note that $i+ j \left\lfloor\frac{n-i+1}{\alpha}\right\rfloor > n+\frac{1}{\alpha}-1$ when $j\ge 1$.
Indeed, 
\[ i+ j \left\lfloor\frac{n-i+1}{\alpha}\right\rfloor  \ge  i+ \left\lfloor\frac{n-i+1}{\alpha}\right\rfloor
\ge i+ \frac{n-i+1}{\alpha} - 1 > i+ n-i+ \frac{1}{\alpha} - 1 = n+\frac{1}{\alpha}-1\,.\] 
Hence, we have
\begin{IEEEeqnarray*}{rCl}
n+ \sum_{\ell=1}^{n}  \left\lfloor \frac{i-\ell+ j \left\lfloor\frac{n-i+1}{\alpha}\right\rfloor}{\left\lfloor \frac{n-\ell+1}{\alpha}\right\rfloor} \right\rfloor  
& \le & n + \sum_{\ell=1}^{n}  \frac{i-\ell+ j \left\lfloor\frac{n-i+1}{\alpha}\right\rfloor}{\left\lfloor \frac{n-\ell+1}{\alpha}\right\rfloor}  \\
& = &  n  + {\textstyle{\left( i+ j \left\lfloor \frac{n-i+1}{\alpha}\right\rfloor\right)}} \sum_{\ell=1}^{n}\frac{1}{\left\lfloor \frac{n-\ell+1}{\alpha}\right\rfloor} - \sum_{\ell=1}^{n}\frac{\ell}{\left\lfloor \frac{n-\ell+1}{\alpha}\right\rfloor}  \\
& \le &  n  + {\textstyle{\left( i+ j \left\lfloor \frac{n-i+1}{\alpha}\right\rfloor\right)}} \sum_{\ell=1}^{n}\frac{1}{\frac{n-\ell+1}{\alpha} - 1} - \sum_{\ell=1}^{n}\frac{\ell}{ \frac{n-\ell+1}{\alpha}}  \\
& \le &  n  + {\textstyle{\left( i+ j \left\lfloor \frac{n-i+1}{\alpha}\right\rfloor\right)} \frac{\alpha}{1-\alpha}} H_n - \alpha\int_{0}^{n}\frac{x}{n-x+1}dx  \\
& \le &  n  + {\textstyle{\left( i+ j \left\lfloor \frac{n-i+1}{\alpha}\right\rfloor\right)}\left( 1 - \frac{1- \alpha -\alpha H_n}{1- \alpha}\right) } - \alpha \left( (n+1)\ln(n+1) - n \right) \\
& \le & {\textstyle{\left( i+ j \left\lfloor \frac{n-i+1}{\alpha}\right\rfloor\right)} +  n   - \frac{1- \alpha -\alpha H_n}{1- \alpha} \left(n+\frac{1}{\alpha}-1\right)} - \alpha \left( (n+1)\ln(n+1) - n \right)
\end{IEEEeqnarray*}
At this point, it suffices to show that  
\[n   - \frac{1- \alpha -\alpha H_n}{1- \alpha} \left(n+\frac{1}{\alpha}-1\right) - \alpha \left( (n+1)\ln(n+1) - n \right)\le 0 .\]
Using $\alpha=\frac{1}{n^{0.5+\varepsilon}}$, the above is equivalent to
\begin{IEEEeqnarray*}{l}
n  - \left(1 - \frac{H_n}{n^{0.5+\varepsilon}-1}\right)  \left(n+n^{0.5+\varepsilon}-1\right) - \frac{(n+1)\ln(n+1)}{n^{0.5+\varepsilon}} + n^{0.5-\varepsilon} \le 0  \\
 \Longleftrightarrow n  - n - n^{0.5+\varepsilon} + 1 + \frac{n H_n}{n^{0.5+\varepsilon}-1} + H_n  -\frac{(n+1)\ln(n+1)}{n^{0.5+\varepsilon}} + n^{0.5-\varepsilon} \le 0 \\
 \Longleftrightarrow  - n^{0.5+\varepsilon}+ H_n + n^{0.5-\varepsilon} + 1 + \frac{n H_n}{n^{0.5+\varepsilon}-1}  -\frac{(n+1)\ln(n+1)}{n^{0.5+\varepsilon}}  \le 0 \,.  
\end{IEEEeqnarray*}
To see that the latter is clearly true for large $n$, notice that 
\begin{IEEEeqnarray*}{rCl}
\frac{n H_n}{n^{0.5+\varepsilon}-1}  -\frac{(n+1)\ln(n+1)}{n^{0.5+\varepsilon}}
& = & \frac{n^{1.5+\varepsilon}(H_n-\ln(n+1))+\left( n + 1 - n^{0.5+\varepsilon} \right) \ln(n+1)}{n^{0.5+\varepsilon}(n^{0.5+\varepsilon}+1)}   \\
& \le &  \frac{n^{1.5+\varepsilon} + n \ln(n+1)}{n^{0.5+\varepsilon}(n^{0.5+\varepsilon}+1)} = \Theta\left(n^{0.5-\varepsilon}\right) \,, 
\end{IEEEeqnarray*}
where we used the known fact that $\ln(n+1)\le H_n\le \ln(n+1) +1$. This proves \eqref{sumoffloors}, and completes the proof. (Again, it is not necessary that $n$ is very large for things to work. E.g., for $\varepsilon=0.25$, it suffices to have $n\ge 17$.)
\end{proof}

\section{A Simple Randomized Mechanism}
\label{sec:probab}
Despite the negative results in the previous sections, there are (asymptotically) good guarantees 
when the values are random and the mechanism is randomized as well. 
In fact, we consider the very simple mechanism that independently allocates each item uniformly at random.
The approximation guarantee for maximin share follows from the corresponding proportionality guarantee in Theorem 
\ref{thm:random-prop} below. The theorem works for a wide range of distributions, including the 
--discrete or continuous-- uniform distribution on subsets of $[0, 1]$. 

The $D_i(n, m)$s, here, 
are distributions over $[0, 1]$ with the following property: there exists some $\varepsilon>0$ such that for 
any $n, m \in \mathbb N$ and any $i\in [n]$ the mean of $D_i(n, m)$ is at least $\varepsilon$. The simplest 
--although quite realistic-- case is when each $D_i$ does not depend on $n$ and $m$ at all, but in general 
this need not be the case, as long as that their means do not vanish as $n$ and $m$ grow large.
Also, notice that independence is only assumed between the values of different items, and therefore the
valuation functions of different players can be correlated.

We should mention that the objective here is different from the objective in the probabilistic analyses of
\cite{AMNS15} and \cite{KPW16}. We do not hope to produce an allocation that gives to each player $i$ value at least $\mms_i$ with high probability; 
each player should only get a fraction of that, subject to truthfulness. This is why we are able to 
cover such a wide range of distributions, using such a naive mechanism.

\begin{theorem}\label{thm:random-prop}
Let $N=[n]$ be a set of players and $M=[m]$ be a set of items, and for each $i\in [n]$ assume that the $v_{ij}$s are 
i.i.d.\ random variables that follow $D_i(n, m)$, where the $D_i(n, m)$s are as described above. 
Then, for any $\rho\in[0, 1)$ and for large enough $m$, there is a truthful randomized mechanism that allocates 
to each player $i$ a set of items of total value at least $\frac{\rho}{n}v_i(M) \ge \rho\cdot\mms_i$ with probability 
$1-O\left(n^2/m\right)$. 
\end{theorem}

\begin{proof}
In what follows we consider the mechanism that independently allocates each item uniformly at random among the players.
Truthfulness follows from the fact that the mechanism completely ignores any input given by the players.

For any $i\in [n], j\in [m]$ let $X_{ij}$ be the indicator r.v. for the event ``player $p_i$ gets item $j$'', 
and $Y_i$ be the total value $p_i$ receives. We have $Y_i=\sum_{j}X_{ij}v_{ij}$. Next, we calculate $\mathrm{E}[Y_i]$
and $\mathrm{Var}(Y_i)$. Clearly, $\mathrm{E}[X_{ij}]=\frac{1}{n}$ and 
$\mathrm{Var}(X_{ij})=\frac{1}{n}-\frac{1}{n^2}=\frac{n-1}{n^2}$.
If by $\mu_i$ and $\sigma_i^2$ we denote the mean and the variance of $D_i(n, m)$ respectively\footnote{Of course, $\mu_i$ and $\sigma_i^2$ are functions of $n$ and $m$, but for simplicity we drop their arguments.}, then 
\[\mathrm{E}[Y_i]=\sum_{j}\mathrm{E}[X_{ij}v_{ij}]= \sum_{j}\frac{1}{n}\mu_i=\frac{m\mu_i}{n}\,,\] 
and 
\begin{IEEEeqnarray*}{rCl}
\mathrm{Var}(Y_i) & = & \sum_{j=1}^{m} \mathrm{Var}(X_{ij}v_{ij})  =  \sum_{j=1}^{m} \left(\mathrm{Var}(X_{ij})\cdot\mathrm{Var}(v_{ij}) +
\mathrm{E}^2[X_{ij}]\cdot\mathrm{Var}(v_{ij}) + \mathrm{Var}(X_{ij})\cdot\mathrm{E}^2[v_{ij}]\right) \\ 
	& = & \sum_{j=1}^{m} \left(\frac{\sigma_i^2(n-1)}{n^2} +
\frac{\sigma_i^2}{n^2} + \frac{(n-1)\mu_i^2}{n^2}\right)  =  \frac{m\left(n\sigma_i^2+(n-1)\mu_i^2\right)}{n^2}\,,
\end{IEEEeqnarray*}
by using the independence of $X_{ij}$ and $v_{ij}$ for any $i\in [n], j\in [m]$, as well as 
the independence of $X_{ij_1}v_{ij_1}$ and $X_{ij_2}v_{ij_2}$ for any $i\in [n], j_1, j_2\in [m]$. Notice that $\sigma_i^2 \le 1-\mu_i^2$ and therefore
\[\mathrm{Var}(Y_i)\le \frac{m\left(n(1-\mu_i^2) +(n-1)\mu_i^2\right)}{n^2} = \frac{m\left(n -\mu_i^2\right)}{n^2} \le \frac{m}{n}  \,.\]

\noindent Now, by setting $a_i =\frac{\rho m \mu_i + \rho m^{0.75}}{n}$, we have
\begin{IEEEeqnarray*}{rCl}
\mathrm{P}\left( \exists i \text{ such that } Y_i < \frac{\rho}{n}v_i(M)\right)  
			&\le& \sum_{i=1}^{n} \mathrm{P} \left( Y_i<\frac{\rho\cdot v_i(M)}{n}\right)  \\ 
&=& \sum_{i=1}^{n} \mathrm{P} \left( Y_i<\min\bigg\lbrace \frac{\rho\cdot v_i(M)}{n}, a_i\bigg\rbrace 
										\text{\ \ or\ \ } \frac{\rho\cdot v_i(M)}{n}>\max\left\lbrace Y_i, a_i\right\rbrace \right)  \\
&\le& \sum_{i=1}^{n} \mathrm{P} \left( Y_i<\min\bigg\lbrace \frac{\rho\cdot v_i(M)}{n}, a_i\bigg\rbrace\right) 
+\sum_{i=1}^{n} \mathrm{P} \left( \frac{\rho\cdot v_i(M)}{n}>\max\left\lbrace Y_i, a_i\right\rbrace \right)  \\
& \le & \sum_{i=1}^{n} \mathrm{P}(Y_i<a_i) + \sum_{i=1}^{n} \mathrm{P}\left( \frac{\rho\cdot v_i(M)}{n}>a_i \right) \,.
\end{IEEEeqnarray*}
To upper bound the first sum, we use Chebyshev's inequality:
\begin{IEEEeqnarray*}{rCl}
\mathrm{P}\left(Y_i<a_i\right)& = & \mathrm{P}\left(\mathrm{E}[Y_i] -Y_i \ge \mathrm{E}[Y_i]-a_i\right)   \le   \mathrm{P}\left(\left|\mathrm{E}[Y_i] -Y_i \right| \ge \frac{m\mu_i}{n}-\frac{\rho m \mu_i + \rho m^{0.75}}{n}\right)\\
	&=& \mathrm{P}\left(\left|\mathrm{E}[Y_i] -Y_i\right| \ge \frac{(1-\rho) m \mu_i - \rho m^{0.75}}{n\cdot \sigma_{Y_i}}\sigma_{Y_i}\right)   \le  \frac{n^2 \sigma_{Y_i}^2}{\left((1-\rho) m \mu_i - \rho m^{0.75}\right)^2} \\
	&\le& \frac{n^2 \frac{m}{n} }{m^2\left((1-\rho)\mu_i - \frac{\rho}{m^{0.25}}\right)^2} = 
	\frac{n}{m\left((1-\rho)\mu_i - \frac{\rho}{m^{0.25}}\right)^2} = O\left(\frac{n}{m}\right)\,,
\end{IEEEeqnarray*}
and thus
\[\sum_{i=1}^{n} \mathrm{P}(Y_i<a_i) = O\left(\frac{n^2}{m}\right)\,.\]
On the other hand, using Hoeffding's inequality,
\begin{IEEEeqnarray*}{rCl}
\mathrm{P}\left( \frac{\rho\cdot v_i(M)}{n}>a_i\right) & = & 
        \mathrm{P}\left( \frac{v_n(M)}{m} -\mu_i > \frac{n\cdot a_i}{\rho\cdot m}-\mu_i\right)   =   \mathrm{P}\left( \frac{v_n(M)}{m} -\mu_i > \frac{\rho m \mu_i + \rho m^{0.75} - \rho m \mu_i}{\rho\cdot m}\right)\\
&\le& e^{-2m\left( \frac{1}{m^{0.25}} \right)^2} \le e^{-2\sqrt{m}}= o\left(\frac{1}{m}\right)\,,
\end{IEEEeqnarray*}
and thus
\[\sum_{i=1}^{n} \mathrm{P}\left( \frac{\rho\cdot v_i(M)}{n}>a_i \right) = o\left(\frac{n}{m}\right)\,.\]

\noindent Finally, we have
\begin{IEEEeqnarray*}{x+C+x*}
& \mathrm{P}\left( \forall i, \ Y_i \ge \frac{ \rho\cdot v_i(M)}{n}\right) \ge 1- \sum_{i=1}^{n} \mathrm{P}(Y_i<a_i) - \sum_{i=1}^{n} \mathrm{P}\left( \frac{\rho\cdot v_i(M)}{n}>a_i \right)  = 1-O\left(\frac{n^2}{m}\right) \,. &\qedhere
\end{IEEEeqnarray*}

\end{proof}

\section{Conclusions}
\label{sec:concl}

We embarked on the existence of truthful mechanisms for approximate maximin share allocations. In doing so, we considered three models regarding the information that a mechanism elicits from the players, and studied their power and limitations. Quite surprisingly, we have exhibited cases with two players where the best possible truthful approximation is achieved by using only ordinal information.

Our work leaves several interesting questions for future research. 
A great open problem is whether there exist better truthful mechanisms in the cardinal model that explicitly take into account the players' valuation functions rather than just ordinal information. Understanding the power of ordinal information is an important direction in our view, which is along the same lines as the work of \cite{AS16} for other optimization problems. The fact that for $n=2$ and $m\in\{4,5\}$, the ordinal and the cardinal model are equivalent was rather surprising to us and it remains open to explore how far this equivalence can go. 
Another more general question is  to tighten the upper and lower bounds obtained here; especially for a large number of players, these bounds are quite loose.

\bibliographystyle{plain}
\bibliography{fairdivrefs}

\appendix

\section{Missing Proofs from Sections \ref{sec:rankings} and \ref{sec:val-rank}}

\vspace{7pt}
\begin{proof}[\textbf{Alternative Proof of Corollary \ref{cor:ordinal-1/2}}]
It suffices to prove the statement when we only have 4 items $a, b, c, d$. Notice that since the mechanism takes as input the players' rankings, to achieve a $(1/2+\varepsilon)$-approximation for $\varepsilon \in (0,1/2]$, there are some allocations that are not feasible. Specifically, the mechanism cannot allocate only one item to one player and three items to the other since there is the possibility that the first player values the items equally. 
Moreover, the mechanism cannot give to $p_i$ the bundles $B_i(2,4)$ or $B_i(3,4)$. To see this, let $\mathbf v_i=[2,1,1,0]$; both $B_i(2,4)$ and  $B_i(3,4)$ have a total value of $1$ while $\mms_i= 2$. Thus, the only feasible bundles that ensure a  $(1/2+\varepsilon)$-approximation for a player in the ordinal model, are $B_i(1,2)$, $B_i(1,3)$, $B_i(1,4)$ and $B_i(2,3)$.  Now suppose that there exists a deterministic truthful mechanism that achieves the desired approximation ratio and consider the following profiles:
\\
\textit{Profile 1}:  $\{[a \succeq_1 b \succeq_1 c \succeq_1 d],  [a \succeq_2 b \succeq_2 c \succeq_2 d]\}$. 
There are two feasible allocations, i) $(\{a,d\}, \{b,c\})$ and ii) $(\{b,c\}, \{a,d\})$. W.l.o.g. we may assume  that the mechanism outputs allocation i). The analysis in the other case is symmetric. 
\\
\textit{Profile 2}:  $\{[a \succeq_1 b \succeq_1 c \succeq_1 d],  [a \succeq_2 b \succeq_2 d \succeq_2 c]\}$.
There are two feasible allocations i) $(\{a,c\}, \{b,d\})$ and ii) $(\{b,c\}, \{a,d\})$.  Allocation ii) is not possible, since in Profile 1 $p_2$  could play $[a \succeq_2 b \succeq_2 d \succeq_2 c]$  and get items $\{a,d\}$ which --depending on $p_2$'s valuation function-- can have strictly more value than $\{b, c\}$. 
Thus, the mechanism outputs allocation i).
\\
\textit{Profile 3}:  $\{[a \succeq_1 b \succeq_1 c \succeq_1 d],  [b \succeq_2 d \succeq_2 c \succeq_2 a]\}$.
There are three feasible allocations, i) $(\{a,c\}, \{b,d\})$, ii) $(\{a,b\}, \{c,d\})$ and iii) $(\{a,d\}, \{b,c\})$. Allocations ii) and iii) are not possible, since in Profile 3 $p_2$ could play  $[a \succeq_2 b \succeq_2 d \succeq_2 c]$ and get items $\{b, d\}$ which might have more value than $\{c,d\}$ or $\{b,c\}$. Thus, the mechanism outputs allocation i).
\\
\textit{Profile 4}:  $\{[b \succeq_1 a \succeq_1 c \succeq_1 d],  [b \succeq_2 d \succeq_2 c \succeq_2 a]\}$.
There are two feasible allocations i) $(\{a,c\}, \{b,d\})$ and ii) $(\{a,b\}, \{c,d\})$.  Allocation ii) is not possible, since in Profile 3 $p_1$  could play $[b \succeq_1 a \succeq_1 c \succeq_1 d]$  and get items $\{a, b\}$ which might have more value than $\{a, c\}$. 
Thus, the mechanism outputs allocation i).
\\
\textit{Profile 5}:  $\{[b \succeq_1 a \succeq_1 c \succeq_1 d],  [b \succeq_2 a \succeq_2 c \succeq_2 d]\}$.
There are two feasible allocations i) $(\{a,c\}, \{b,d\})$ and ii) $(\{b,d\}, \{a,c\})$.  Allocation ii) is not possible, since in Profile 5 $p_2$  could play $[b \succeq_2 d \succeq_2 c \succeq_2 a]$  and get items $\{b,d \}$ which might have more value than $\{a, c\}$. 
Thus, the mechanism outputs allocation i).
\\
\textit{Profile 6}:  $\{[b \succeq_1 a \succeq_1 c \succeq_1 d],  [a \succeq_2 b \succeq_2 c \succeq_2 d]\}$.
There are two feasible allocations i) $(\{b,c\}, \{a,d\})$ and ii) $(\{b,d\}, \{a,c\})$.  Allocation ii) is not possible, since in Profile 5 $p_2$  could play $[a \succeq_2 b \succeq_2 c \succeq_2 d]$  and get items $\{a,c\}$ which might have more value than $\{b, d\}$. Allocation i) is not possible either, since in Profile 1 $p_1$  could play $[b \succeq_1 a \succeq_1 c \succeq_1 d]$  and get items $\{b, c\}$ which might have more value than $\{a,d\}$.
Thus, there is no possible allocation in this profile, which is a contradiction.
\end{proof}

\vspace{10pt}
\begin{proof}[\textbf{Proof of Theorem \ref{thm:5/6} for \boldmath${m=5}$}]
We shall study two cases, of  five profiles each. In both cases and for all profiles, the ordering of the items is $a\succeq_i b\succeq_i c\succeq_i d\succeq_i e $  for both players. Suppose that there exists a deterministic truthful mechanism that achieves a $(5/6+\varepsilon)$-approximation for some $\varepsilon \in (0,1/5]$.
\\
\textit{Profile 1:} $\{[1, 1, 1, 1, 1],[1, 1, 1, 1, 1]\}$. Here,  $\mms_i =2$ for $i\in\{1, 2\}$. The mechanism must give to each player items of value greater than $5/6\cdot 2 = 1.67$. Thus each player has to receive at least two items. W.l.o.g. we may assume that $p_1$ gets three items and $p_2$ gets two items (the analysis in the other case is symmetric). There are two cases to be considered depending on who gets item $a$:
\\
\textbf{Case 1:} $p_1$ gets item $a$ in Profile 1.
\\
\textit{Profile 2:} $\{[1, 0.25, 0.25, 0.25, 0.25],[1,1,1,1,1]\}$. Here $\mms_1 =1$  and $\mms_2 =2$. The mechanism must give to $p_1$ a total value greater than $5/6 \cdot 1=0.83$ and to $p_2$ a total value greater than $5/6\cdot 2= 1.67$. Notice now that $p_2$ has to get at least two items, thus $p_1$ has to get item $a$. In addition, $p_1$ gets three items, or she could play $\mathbf v_1'=[1, 1, 1, 1, 1]$ and end up strictly better. So, we  conclude that in this profile $p_1$ gets three items, with $a$ among them, while $p_2$ gets two items.
\\
\textit{Profile 3:} $\{[1, 0.25, 0.25, 0.25, 0.25],[1, 0.25, 0.25, 0.25, 0.25]\}$. Here $\mms_i =1$ for $i\in\{1, 2\}$, so in order to achieve something greater than $5/6\cdot 1= 0.83$, there are only two feasible allocations: i) $p_2$ gets $a$ and $p_1$ gets the remaining items and ii) $p_1$ gets $a$ and $p_2$ gets the remaining items. Notice that allocation ii) is not possible, since $p_2$ in Profile 2 could play $\mathbf v_2'= [1,0.25,0.25,0.25,0.25]$  and be better off. So, here the mechanism outputs allocation i).
\\
\textit{Profile 4:} $\{[1, 1, 1, 1, 1],[1,0.25,0.25,0.25,0.25]\}$. Here $\mms_1 =2$  and $\mms_2 =1$. The mechanism must give to $p_1$ a total value greater than $5/6 \cdot 2=1.67$ and to $p_2$ a total value greater than $5/6\cdot 1= 0.83$. Notice now that $p_1$ must get four items, or otherwise in Profile 4 she could play $\mathbf v_1'=[1, 0.25, 0.25, 0.25, 0.25]$ and end up strictly better. Thus $p_2$ can get only one item and this must be $a$, since it is the only item that achieves the desired ratio.
\\
\textit{Profile 5:} $\{[1, 1, 1, 1, 1],[0.55, 0.45, 0.34, 0.34, 0.34]\}$. Here $\mms_1 =2$  and $\mms_2 =1$. The mechanism must give to $p_1$ a total value greater than $5/6 \cdot 2=1.67$ and to $p_2$ a total value greater than $5/6\cdot 1= 0.83$. First notice that $p_1$ must get at least two items. Given that, let us now examine the feasible bundles for $p_2$: i) $\{a,b\}$, ii) $\{a,c\}$, iii) item $a$ and two more items, iv) item $b$ and two more items, excluding item $a$, and v) $\{c,d,e\}$. Bundles i), ii), iii) are not possible, since $p_2$ in Profile 4 could play $\mathbf v_2'= [0.55,0.45,0.34,0.34,0.34]$ and end up strictly better. Allocations iv), v) are not possible either, since $p_2$ in Profile 1 could  again play $\mathbf v_2'=[0.55,0.45,0.34,0.34,0.34]$ and end up strictly better. Thus, we can conclude that there is no possible allocation in this profile, which leads to a contradiction.
\\
\textbf{Case 2:} $p_2$ gets item $a$ in Profile 1.
\\
\textit{Profile 2:} $\{[1, 1, 1, 1, 1],[1,0.25,0.25,0.25,0.25]\}$. Here $\mms_1 =2$  and $\mms_2 =1$. The mechanism must give to $p_1$ a total value greater than $5/6 \cdot 1=0.83$ and to $p_2$ a total value greater than $5/6\cdot 2= 1.67$. Notice now that $p_2$ must get two items, including item $a$. Indeed, if she gets three items, then in Profile 1 she could play $\mathbf v_2'=[1,0.25,0.25,0.25,0.25]$ and end up better off. If she gets one item (item $a$), then in Profile 2 she could play $\mathbf v_2''=[1,1,1,1,1]$ and end up strictly better. So, we  conclude that here $p_1$ gets three items and $p_2$ gets two items, with $a$ among them.
\\
\textit{Profile 3:} $\{[1, 0.25, 0.25, 0.25, 0.25],[1,0.25,0.25,0.25,0.25]\}$.  Here $\mms_i =1$ for $i\in\{1, 2\}$, so in order to achieve something higher than $5/6\cdot 1= 0.83$, there are only two feasible allocations: i) $p_2$ gets item $a$ and $p_1$ gets the remaining items, and ii) $p_1$ gets item $a$ and $p_2$ gets the remaining items. Notice that allocation i) is not possible, since  $p_1$ in Profile 2 could play $\mathbf v_1'=[1, 0.25, 0.25, 0.25, 0.25]$ and end up strictly better. So, the mechanism outputs allocation ii).
\\
\textit{Profile 4:} $\{[1, 0.25, 0.25, 0.25, 0.25],[0.5,0.5,0.35,0.33,0.32]\}$. Here $\mms_i =1$ for $i\in\{1, 2\}$, so the mechanism must give to both players bundles of value greater than $5/6\cdot 1= 0.83$. Notice that $p_2$ must get at least two items and thus, $p_1$ must get item $a$ to achieve the desired ratio. However, if $p_2$ gets less than four items, then in Profile 4 she  could play $\mathbf v_2'=[1,0.25,0.25,0.25,0.25]$  and end up strictly better. Thus, we can conclude that $p_1$ gets item $a$ and $p_2$ gets the remaining items.
\\
\textit{Profile 5:} $\{[0.5, 0.2, 0.2, 0.2, 0.1],[0.5,0.5,0.35,0.33,0.32]\}$. Here $\mms_1 =0.6$  and $\mms_2 =1$. The mechanism must give to $p_1$ a total value greater than $5/6 \cdot 0.6=0.5$ and to $p_2$ a total value greater than $5/6\cdot 1= 0.83$. First notice that in order to achieve the desired ratio $p_1$ must get at least two items. However, she can not get a proper superset of $\{a\}$, since in Profile 4  she could play $\mathbf v_1'=[0.5, 0.2, 0.2, 0.2, 0.1]$ and end up strictly better. The only remaining feasible bundles for $p_1$ are: i)  $\{b,c,d\}$ and ii)  $\{b,c,d,e\}$. It is easy to see that none of these  are possible, since for the allocation implied by i) $p_2$ gets a total value of $0.82$, and for the allocation implied by  ii) $p_2$ gets a total value of $0.5$. Thus, there is no possible allocation in this profile, which leads to a contradiction.
\end{proof}

\end{document}